\documentclass[a4paper,11pt]{article}
\usepackage[T1]{fontenc}\selectfont

\usepackage{geometry}
\usepackage{hyperref}
\usepackage[fleqn]{amsmath}

\usepackage{amsmath,amsthm,amsfonts}
\usepackage{graphicx,float,caption}
\usepackage{braket,mathtools,scalefnt,anyfontsize}

\usepackage{shadow}
\renewcommand{\=}{:=}
\newcommand{\be}{\beta}
\newcommand{\p}{\partial}
\newcommand{\pT}{\partial^{T}}
\newcommand{\ZZ}{\mathbb{C}}

\newtheorem{thm}{Theorem}[section]

\numberwithin{figure}{section}
\numberwithin{equation}{section}
\usepackage{cleveref}

\begin{document}

\allowdisplaybreaks

\title{\Large\bf Note on star-triangle equivalence\\in conducting networks}
\author{E Paal and M Umbleja}
\date{}

\maketitle

\begin{abstract} 
By using the discrete Poisson equations the star-triangle (external) equivalence in conducting networks is considered and the Kennelly famous transformation formulae [Kennelly A E 1899 Electrical World and Engineer {\bf34} 413] are explicitly restated.
\end{abstract}

%\address{Tallinn University of Technology, Ehitajate tee 5, 19086 Tallinn, Estonia}

\thispagestyle{empty}

\section{Introduction and outline of the paper} 

The homological representation and modeling \cite{Roth55} of networks (n/w) is based on their geometric elements, called also the chains -- nodes, branches (edges), meshes (simple closed loops), and using the natural geometric boundary operator of the n/w which only depends on the geometry (topology) of the n/w. Then, both of the Kirchhoff laws can be presented in a compact algebraic form that may be called the homological Kirchhoff Laws.

In the present note, we compose the discrete Poisson equations and consider the star-triangle (external) equivalence transformation in conducting networks, see Fig. \ref{fig_star_triangle}, and prove the Kennelly famous transformation formulae \cite{Kennelly}. We use the geometrical representation that was explained in  \cite{JPCS}.
\begin{figure}[h]
\begin{center}
\includegraphics[height=40mm]{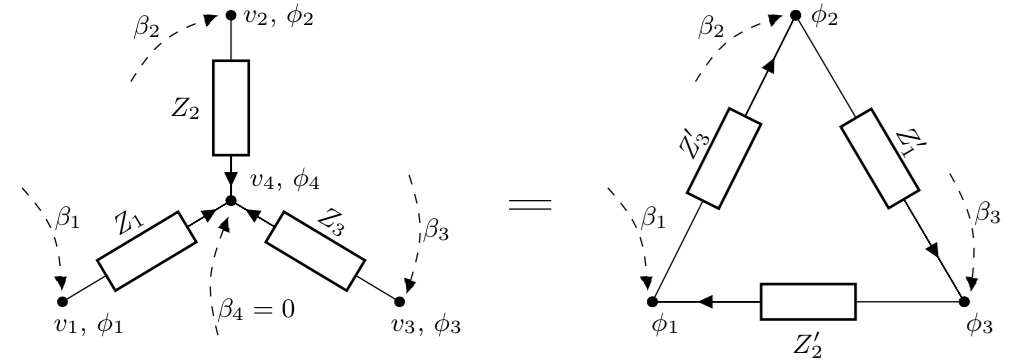}
\caption{Star-triangle (external/boundary) equivalence\label{fig_star_triangle}}
\end{center}
\end{figure}

First introduce some notations, here we follow \cite{JPCS}. 
The both circuits are assumed to have the same boundary conditions. We denote 
\begin{subequations}
\begin{align}
\ket{\be}
&\=
\ket{\be_1 \be_2 \be_3 \be_4}
,\quad
\ket{\be}'
\=
\ket{\be_1 \be_2 \be_3}
\quad \text{(boundary currents)} \\
\ket{\phi}
&\=
\ket{\phi_1 \phi_2 \phi_3 \phi_4}
,\quad \!
\ket{\phi}'
\=
\ket{\phi_1 \phi_2 \phi_3}
\quad \text{(node potentials)}
\end{align}
\end{subequations}
The impedance matrices are
\begin{align}
Z 
\=
\begin{bmatrix*}[r]
Z_{1} & 0 & 0\\
0 & Z_{2} & 0\\
0 & 0 & Z_3
\end{bmatrix*}
,
\quad 
Z' 
\=
\begin{bmatrix*}[r]
Z'_{3} & 0 & 0\\
0 & Z'_{1} & 0\\
0 & 0 & Z'_2
\end{bmatrix*}
\end{align}
The admittances $Y_n$ and $Y'_n$ are defined by
\begin{align}
Y_n Z_n = 1 = Y'_n Z'_n
,
\quad
n=1,2,3
\end{align}
and the admittance matrices are
\begin{align}
Y
\=
Z^{-1}
=
\begin{bmatrix*}[r]
Y_{1} & 0 & 0\\
0 & Y_{2} & 0\\
0 & 0 & Y_3
\end{bmatrix*}
,
\quad 
Y'
\=
Z'^{-1}
=
\begin{bmatrix*}[r]
Y'_{3} & 0 & 0\\
0 & Y'_{1} & 0\\
0 & 0 & Y'_2
\end{bmatrix*}
\end{align}
  
\section{Star}

Consider the star circuit represented on Fig. \ref{fig_star_triangle}. Define
\begin{itemize}
\itemsep-2pt
\item
\emph{Node space} $C_0\=\braket{v_1 v_2 v_3 v_4}_\ZZ$, \quad $\dim C_0=4$
\item
\emph{Branch space} $C_1\=\braket{e_1 e_2 e_3}_\ZZ$, \quad $\dim C_1=3$
\end{itemize}
First construct the boundary operator $\p : C_1\to C_0$. By definition,
\begin{subequations}
\begin{align}
\p  e_1 = \p (v_1 v_4) \= v_4-v_1 =: \ket{-1; 0 ;0;1}\\
\p  e_2 = \p (v_2 v_4) \= v_4-v_2 =: \ket{0;-1; 0;1}\\
\p  e_3 = \p (v_3 v_4)\= v_4-v_3 =: \ket{ 0 ;0;-1;1}
\end{align}
\end{subequations}
and in the matrix representation we have 
\begin{align}
\p  =
\begin{bmatrix*}[r]
-1 & 0 & 0\\
0  & -1 & 0\\
0 & 0 & -1\\
1 & 1& 1
\end{bmatrix*}
\quad\Longrightarrow\quad
\pT   =
\begin{bmatrix*}[r]
-1 & 0 & 0 & 1\\
0 & -1 & 0 & 1\\
0 & 0& -1& 1
\end{bmatrix*}
\end{align}
Now it is easy to calculate the Laplacian as follows:
\begin{subequations}
\begin{align}
\Delta
&\=
\p  Y\pT  
\\
&\hphantom{:}=
\begin{bmatrix*}[r]
-1 & 0 & 0\\
0  & -1 & 0\\
0 & 0 & -1\\
1 & 1& 1
\end{bmatrix*}
\begin{bmatrix*}[r]
Y_{1} & 0 & 0\\
0 & Y_{2} & 0\\
0 & 0 & Y_{3}
\end{bmatrix*}
\begin{bmatrix*}[r]
-1 & 0 & 0 & 1\\
0 & -1 & 0 & 1\\
0 & 0& -1  & 1
\end{bmatrix*} 
\\
&\hphantom{:}= 
\begin{bmatrix*}[r]
Y_1 & 0 & 0 & -Y_1 \\
0 & Y_2 & 0 & -Y_2 \\
0 & 0 & Y_3 & -Y_3 \\
-Y_1 & -Y_2 & -Y_3 & Y_1+Y_2+Y_3
\end{bmatrix*}
\end{align}
\end{subequations}
The Poisson equation 
\begin{align}
\Delta\ket{\phi}
=-\ket{\be}
\end{align}
in coordinate form reads 
\begin{equation}
\begin{cases}
\be_1=\dfrac{-\phi_{1}+\phi_{4}}{Z_{1}}\\[1em]
\be_2=\dfrac{-\phi_{2}+\phi_{4}}{Z_{2}}\\[1em]
\be_3=\dfrac{-\phi_{3}+\phi_{4}}{Z_{3}}\\[1em]
\be_4=-\left(\dfrac{-\phi_{1}+\phi_{4}}{Z_{1}}
+\dfrac{-\phi_{2}+\phi_{4}}{Z_{2}}
+\dfrac{-\phi_{3}+\phi_{4}}{Z_{3}}\right)
\end{cases}
\end{equation}
We can easily check consistency:
\begin{align}
\be_1+\be_2+\be_3+\be_4=0
\end{align}
For $\be_4=0$ we have
\begin{align}
\dfrac{-\phi_{1}+\phi_{4}}{Z_{1}}+
\dfrac{-\phi_{2}+\phi_{4}}{Z_{2}}+
\dfrac{-\phi_{3}+\phi_{4}}{Z_{3}}=0
\end{align}
from which it follows that
\begin{align}
\phi_{4}
= \dfrac{\phi_{1}Y_{1}+\phi_{2}Y_{2}+\phi_{3}Y_{3}}{Y_{1}+Y_{2}+Y_{3}}
\label{neutral}
\end{align}

\section{Triangle}

Next consider the the triangle circuit on Fig. \ref{fig_star_triangle}.
We denote the spanning nodes and branches by the same letters. Then the linear spans are 
\begin{itemize}
\itemsep-2pt
\item
\emph{Node space} $C_0\=\braket{v_1 v_2 v_3}_\ZZ,\quad\dim C_0=3$
\item
\emph{Branch space} $C_1\=\braket{e_1 e_2 e_3}_\ZZ,\quad\dim C_1=3$
\end{itemize}
Construct the boundary operator  $\p : C_1\to C_0$.
We can see that 
\begin{align}
\p  e_1 = \p (v_1 v_2) \= v_2-v_1 =: \ket{-1; 1 ;0}\\
\p  e_2 = \p (v_2 v_3) \= v_3-v_2 =: \ket{0;-1;1}\\
\p  e_3 = \p (v_3 v_1) \= v_1-v_3 =: \ket{ 1;0;-1}
\end{align}
and the matrix representation is 
\begin{align}
\p  =
\begin{bmatrix*}[r]
-1 & 0 & 1\\
1  & -1 & 0\\
0 & 1 & -1
\end{bmatrix*}
\quad\Longrightarrow\quad
\pT   =
\begin{bmatrix*}[r]
-1 & 1 & 0\\
0 & -1 & 1\\
1 & 0& -1
\end{bmatrix*}
\end{align}
The Laplacian is
\begin{subequations}
\begin{align}
\Delta' 
&\=
\p  Y' \pT   
\\
&\hphantom{:}=
\begin{bmatrix*}[r]
-1 & 0 & 1\\
1  & -1 & 0\\
0 & 1 & -1
\end{bmatrix*}
\begin{bmatrix*}[r]
Y'_{3} & 0 & 0\\
0 & Y'_{1} & 0\\
0 & 0 & Y'_{2}
\end{bmatrix*}
\begin{bmatrix*}[r]
-1 & 1 & 0\\
0 & -1 & 1\\
1 & 0& -1
\end{bmatrix*} 
\\
&\hphantom{:}=
\begin{bmatrix*}[r]
Y'_3+Y'_2 & - Y'_3 & -Y_2 \\
-Y'_3 & Y'_3+Y'_1 & -Y'_1\\
-Y'_2 & -Y'_1 & Y'_1+Y'_2
\end{bmatrix*}
\end{align}
\end{subequations}
The Poisson equation is 
\begin{align}
\Delta'\ket{\phi}'
=-\ket{\be}'
\end{align}
Hence we have 
\begin{equation}
\begin{cases}
\be_1=\hphantom{-}\dfrac{-\phi_{1}+\phi_{2}}{Z'_{3}}-\dfrac{\phi_{1}-\phi_{3}}{Z'_{2}}\\[1em]
\be_2=-\dfrac{-\phi_{1}+\phi_{2}}{Z'_{3}}+\dfrac{-\phi_{2}+\phi_{3}}{Z'_{1}}\\[1em]
\be_3=-\dfrac{-\phi_{2}+\phi_{3}}{Z'_{1}}+\dfrac{\phi_{1}-\phi_{3}}{Z'_{2}}
\end{cases}
\end{equation}
Check the consistency:
\begin{align}
\be_{1}+\be_{2}+\be_{3}=0
\end{align}

\section{Equivalence}

Now consider the star-triangle equivalence as exposed on Fig. \ref{fig_star_triangle} and prove the Kennelly theorem.

\begin{thm}[A. E. Kennelly \cite{Kennelly}]
If the (external/boundary) equivalence presented on Fig.~\ref{fig_star_triangle} holds, then one has
%\begin{subequations}
\begin{align}
\shabox{$Z_{n}Z'_{n}
=
Z'_{1}Z'_{2}Z'_{3}/(Z'_{1}+Z'_{2}+Z'_{3})
%\\
=
Z_{1}Z_{2}+Z_{2}Z_{3}+Z_{3}Z_{1},\quad n=1,2,3$}
\end{align}
%\end{subequations}
\end{thm}

\begin{proof}

As soon as the boundary currents on Fig. \ref{fig_star_triangle} are considered the same, then we have
\begin{subequations}
\begin{align}
\dfrac{-\phi_{1}+\phi_{4}}{Z_{1}}
&=
\hphantom{+}\dfrac{-\phi_{1}+\phi_{2}}{Z'_{3}}
 -\dfrac{\hphantom{-}\phi_{1}-\phi_{3}}{Z'_{2}} \quad 
 | \cdot Z_{1}
\\
\dfrac{-\phi_{2}+\phi_{4}}{Z_{2}}
&=
-\dfrac{-\phi_{1}+\phi_{2}}{Z'_{3}}
+\dfrac{-\phi_{2}+\phi_{3}}{Z'_{1}} \quad
 | \cdot Z_{2}
\\
\dfrac{-\phi_{3}+\phi_{4}}{Z_{3}}
&=
-\dfrac{-\phi_{2}+\phi_{3}}{Z'_{1}}
+\dfrac{\hphantom{-}\phi_{1}-\phi_{3}}{Z'_{2}} \quad
 | \cdot Z_{3}
\end{align}
\end{subequations}
where $\phi_4$ is given by \eqref{neutral}. 
Hence we obtain equations for the potentials $\phi_1,\phi_2,\phi_3$,  
\begin{subequations}
\begin{align}
-\phi_{1}+\phi_{4}
&=
\phantom{-}(-\phi_{1}+\phi_{2})\dfrac{Z_{1}}{Z'_{3}}
-(\hphantom{-}\phi_{1}-\phi_{3})\dfrac{Z_{1}}{Z'_{2}}
\\
-\phi_{2}+\phi_{4}
&=-
(-\phi_{1}+\phi_{2})\dfrac{Z_{2}}{Z'_{3}}
+(-\phi_{2}+\phi_{3})\dfrac{Z_{2}}{Z'_{1}}
\\
-\phi_{3}+\phi_{4}
&=-
(-\phi_{2}+\phi_{3})\dfrac{Z_{3}}{Z'_{1}}
+(\hphantom{-}\phi_{1}-\phi_{3})\dfrac{Z_{3}}{Z'_{2}}
\end{align}
\end{subequations}
By eliminating here the potential $\phi_{4}$, we get relations for the boundary potentials, 
\begin{subequations}
\begin{align}
-\phi_{1}+\phi_{2}
&=
\hphantom{-}(-\phi_{1}+\phi_{2})\dfrac{Z_{1}}{Z'_{3}}
-(\hphantom{-}\phi_{1}-\phi_{3})\dfrac{Z_{1}}{Z'_{2}}
+(-\phi_{1}+\phi_{2})\dfrac{Z_{2}}{Z'_{3}}
-(-\phi_{2}+\phi_{3})\dfrac{Z_{2}}{Z'_{1}}
\\
-\phi_{2}+\phi_{3}
&=
-(-\phi_{1}+\phi_{2})\dfrac{Z_{2}}{Z'_{3}}
+(-\phi_{2}+\phi_{3})\dfrac{Z_{2}}{Z'_{1}}
+(-\phi_{2}+\phi_{3})\dfrac{Z_{3}}{Z'_{1}}
-(\hphantom{-}\phi_{1}-\phi_{3})\dfrac{Z_{3}}{Z'_{2}}
\\
-\phi_{3}+\phi_{1}
&=
-(-\phi_{2}+\phi_{3})\dfrac{Z_{3}}{Z'_{1}}
+(\hphantom{-}\phi_{1}-\phi_{3})\dfrac{Z_{3}}{Z'_{2}}
-(-\phi_{1}+\phi_{2})\dfrac{Z_{1}}{Z'_{3}}
+(-\phi_{1}+\phi_{3})\dfrac{Z_{1}}{Z'_{2}}
\end{align}
\end{subequations}
%% 
%% 1st
We can easily check consistency of the last Eqs, by summing these we easily obtain $0=0$.
This means that one equation is a linear combination of others and we can variate the independent potentials $\phi_1,\phi_1,\phi_3$ only in two equations. We use the first two Eqs.

By variating  the independent potentials $\phi_1,\phi_1,\phi_3$ and setting  the nontrivial potential  
$\phi_3=1$ in the first equation we obtain
\begin{align}
0
=
\dfrac{Z_{1}}{Z'_{2}}-\dfrac{Z_{2}}{Z'_{1}}
\quad\Longrightarrow\quad 
\shabox{$Z_{1}Z'_{1}=Z_{2}Z'_{2}$}
\end{align}
Now take $\phi_1=1$,
\begin{subequations}
\begin{align}
1=\dfrac{Z_{1}}{Z'_{3}}
+\dfrac{Z_{1}}{Z'_{2}}
+\dfrac{Z_{2}}{Z'_{3}}
\quad\Longrightarrow\quad 
1
&=\dfrac{Z_{1}Z'_{2}+Z_{1}Z'_{3}+Z_{2}Z'_{2}}{Z'_{2}Z'_{3}}\\
&=\dfrac{Z_{1}Z'_{2}+Z_{1}Z'_{3}+Z_{1}Z'_{1}}{Z'_{2}Z'_{3}}\\
&=\dfrac{Z_{1}(Z'_{2}+Z'_{3}+Z'_{1})}{Z'_{2}Z'_{3}}
\quad\Longrightarrow\quad
\shabox{$Z_{1}=\dfrac{Z'_{2}Z'_{3}}{Z'_{2}+Z'_{3}+Z'_{1}}$}
\end{align}
\end{subequations}
Next take $\phi_2=1$, 
\begin{subequations}
\begin{align}
1=\dfrac{Z_{1}}{Z'_{3}}+\dfrac{Z_{2}}{Z'_{3}}+\dfrac{Z_{2}}{Z'_{1}}
\quad\Longrightarrow\quad
1
&=\dfrac{Z_{1}Z'_{1}+Z_{2}Z'_{1}+Z_{2}Z'_{3}}{Z'_{3}Z'_{1}}\\
&=\dfrac{Z_{2}Z'_{2}+Z_{2}Z'_{1}+Z_{2}Z'_{3}}{Z'_{3}Z'_{1}}\\
&=\dfrac{Z_{2}(Z'_{2}+Z'_{1}+Z'_{3})}{Z'_{3}Z'_{1}}
\quad\Longrightarrow\quad
\shabox{$Z_{2}=\dfrac{Z'_{1}Z'_{3}}{Z'_{2}+Z'_{1}+Z'_{3}}$}
\end{align}
\end{subequations}
%%
%% 2nd

By variating  the independent potentials $\phi_1,\phi_1,\phi_3$ in the second equation and setting there the nontrivial potential  $\phi_1=1$, we obtain
\begin{align}
0=\dfrac{Z_{2}}{Z'_{3}}-\dfrac{Z_{3}}{Z'_{2}}
\quad\Longrightarrow\quad 
\shabox{$Z_{2}Z'_{2}=Z_{3}Z'_{3}$}
\end{align}
By setting $\phi_3=1$, we obtain
\begin{subequations}
\begin{align}
1=\dfrac{Z_{2}}{Z'_{1}}+\dfrac{Z_{3}}{Z'_{1}}+\dfrac{Z_{3}}{Z'_{2}}
\quad\Longrightarrow\quad 
1
&=\dfrac{Z_{2}Z'_{2}+Z_{3}Z'_{2}+Z_{3}Z'_{1}}{Z'_{1}Z'_{2}}\\
&=\dfrac{Z_{3}Z'_{3}+Z_{3}Z'_{2}+Z_{3}Z'_{1}}{Z'_{1}Z'_{2}}\\
&=\dfrac{Z_{3}(Z'_{3}+Z'_{2}+Z'_{1})}{Z'_{1}Z'_{2}}
\quad\Longrightarrow\quad
\shabox{$Z_{3}=\dfrac{Z'_{1}Z'_{2}}{Z'_{3}+Z'_{2}+Z'_{1}}$}
\end{align}
\end{subequations}
One can easily check that other variations of the potentials $\phi_n$ $(n=1,2,3)$ do not produce additional constraints.
\end{proof}

{\small The research was in part supported by the Estonian Research Council, grant ETF-9038.}

\medskip
{\small
Tallinn University of Technology, Ehitajate tee 5, 19086 Tallinn, Estonia
}

\end{document}